\documentclass[11pt]{article}
\pdfoutput=1

\usepackage{qcircuit}

\usepackage{amsmath} 
\usepackage{amsthm}  
\usepackage{amssymb} 
\usepackage{bold-extra}
\usepackage{multirow}
\usepackage{graphicx}
\usepackage{subcaption}
\usepackage{enumerate}
\usepackage{color}

\theoremstyle{plain}
\newtheorem{theorem}{Theorem} 
\newtheorem{lemma}{Lemma}     
\newtheorem{definition}{Definition} 
\newtheorem{corollary}{Corollary} 

\theoremstyle{definition}

\newcommand{\ket}[1]{|#1\rangle}
\newcommand{\bra}[1]{\langle #1|}
\newcommand{\ketbra}[2]{|#1\rangle\langle#2|}

\newcommand{\hgatematrix}{
  \frac{1}{\sqrt{2}} \begin{bmatrix}
                       1  &  1 \\
                       1  & -1
                     \end{bmatrix}
}
\newcommand{\sgatematrix}{
  \begin{bmatrix}
    1 & 0 \\
    0 & i
  \end{bmatrix}
}
\newcommand{\cnotgatematrix}{
  \begin{bmatrix}
    1 & 0 & 0 & 0 \\
    0 & 1 & 0 & 0 \\
    0 & 0 & 0 & 1 \\
    0 & 0 & 1 & 0
  \end{bmatrix}
}

\newcommand{\cnot}{\mathit{CNOT}}
\newcommand{\swap}{\mathit{SWAP}}

\everymath{\displaystyle}

\begin{document}

\title{Two-Qubit Stabilizer Circuits with Recovery I: Existence}
\author{Wim van Dam\footnote{Department of Computer Science,
Department of Physics, University of California, Santa Barbara, CA, USA}
\and
Raymond Wong\footnote{Department of Computer Science,
University of California, Santa Barbara, CA, USA}}

\maketitle

\begin{abstract}
Understanding how a stabilizer circuit responds to different input is
important to formulating an effective strategy for resource management.
In this paper, we further investigate the many ways of using stabilizer
operations to generate a single qubit output from a two-qubit state.
In particular, by restricting the input to certain product states, we
discover probabilistic operations capable of transforming stabilizer
circuit outputs back into stabilizer circuit inputs. These secondary
operations are ideally suited for recovery purposes and require only
one extra resource ancilla to succeed. As a result of reusing qubits
in this manner, we present an alternative to the original state
preparation process that can lower the overall costs of executing
any two-qubit stabilizer procedure involving non-stabilizer resources.
\end{abstract}


\section{Introduction}
There has been significant progress to building quantum computers.
We can protect qubits with quantum codes, and we can combat the
spread of errors with fault-tolerance; high thresholds approaching
$1\%$ \cite{knill.noisy} is already within reach. Rather, one of
the central challenges is in the efficient handling of noise, where
it is necessary to strike a delicate balance between quality and cost.
Currently many physical qubits are required to achieve this desired
level of protection on a logical qubit \cite{fowler.surface}, but this
comprises only one part of a larger problem. The fact remains that most
fault-tolerant schemes are constrained to a finite number of native
operations, so there is a limit to the type of computations that we
can perform. This usually consists of stabilizer operations --
Clifford group unitaries, Pauli measurements, and ancilla $\ket{0}$
preparation -- which are efficiently simulable on classical computers
and capable of producing highly entangled states. Unfortunately,
stabilizer operations by themselves are not universal, placing a
premium on any non-stabilizer resource added to a circuit.

Magic state distillation is one solution addressing this inherent
limitation of stabilizer operations \cite{bravyi.kitaev.msd}. It
works as follows: prepare imperfect ``magic states,'' measure certain
stabilizer code syndrome operators, then postselect on some target
outcome. The process is repeated recursively until the qubits are at
a high enough quality to consume: the magic states are injected into
quantum circuits to implement quantum gates outside the Clifford group
of operations. The primary concern then is the large input-to-output
ratio before the magic state error is low enough to be permitted in a
fault-tolerant scheme. Numerous proposals over the last few years have
made tremendous advances to increase the efficiencies of this technique
\cite{bravyi.haah.msd, campbell.ogor.magic, dcp.complex, Haah2017,
meier.four.qubit.code}, although the overall format more or less remains
the same. Interestingly, stabilizer operations are enough to perform the
distillation, which is a testament to their versatility. Then given a
supply of non-Clifford gates, we may employ any number of pre-existing
synthesis algorithms to approximate unitaries over this basis. Previous
work has already succeeded in producing solutions able to generate
sequences for single qubit rotations in an optimal fashion
\cite{kmm.asym.opt, kmm.fast.eff, ross.opt.vbasis.zapprox,
ross.sel.opt.tbasis.zapprox}. A recent one even suggests a kind of
distill-and-synthesis hybrid to reduce resource usage even further:
a factor of $3$ savings with quadratic error suppression is possible
over traditional distill-then-synthesize methods
\cite{campbell.howard.unify.1, campbell.howard.unify.2}.

The creativity that went into designing these distillation protocols
is one reason motivating our study of stabilizer operations. Other
uses include procedures for distilling multiple types of magic qubits
\cite{campbell.ogor.magic, dcp.complex, Haah2017, landahl.ces.2013},
as well as implementing phase rotations with low depth circuits.
Some notable examples of the latter are contained in \cite{dcs} and
\cite{jones.qchem}, both of which feature the same stabilizer circuit
to perform the operation. The differences lie in the pre-computed
ancillas injected into the circuit, where Duclos-Cianci and Svore
\cite{dcs} additionally demonstrated how to use the same circuit to
create other resource qubits. At any rate, though simple, both displayed
the advantages of having a large set of non-homogenous states at our
disposal, and all that is required is a two-qubit stabilizer circuit.

Here we consider general two-to-one stabilizer procedures that take
a two-qubit state and produces a single qubit output using stabilizer
operations only. Our intent is to explore these processes from a
different angle, outside the realm of state distillation, and simply
examine their behavior on more arbitrary input. Some limits on distilling
two-qubit states are already discussed in \cite{reichardt.msd.2006}.
Instead, we refine the implementation details first provided by Reichardt
\cite{reichardt.msd.2006} to identify three circuit configurations
characterizing all such procedures. These three forms suggest that in
addition to Pauli measurements and postselection, single qubit Clifford
gates and at most one $\cnot$ or $\swap$ are enough to realize any
stabilizer procedure acting on two qubits. When the input set is further
confined to certain product states, we discover an interesting connection
between stabilizer circuits of the single $\cnot$ variety: there are
``recovery circuits'' that can recuperate a product state input from
a corrupted stabilizer circuit output. For our main result, we show that
\emph{any two-to-one stabilizer procedure realizable by one $\cnot$ has
a recovery circuit, and that every recovery circuit to this procedure has
the same probability of recovery}, which means any will suffice to install
as part of a larger computation. We end with a few experiments showcasing
the benefits of recovery circuits.

\section{Preliminaries}
This section provides an overview of the elementary stabilizer
operations and basic concepts.  The single qubit Pauli matrices
are
\begin{align}
I = \begin{bmatrix} 1 &  0 \\ 0 &  1 \end{bmatrix}, &&
X = \begin{bmatrix} 0 &  1 \\ 1 &  0 \end{bmatrix}, &&
Y = \begin{bmatrix} 0 & -i \\ i &  0 \end{bmatrix}, &&
Z = \begin{bmatrix} 1 &  0 \\ 0 & -1 \end{bmatrix}.
\end{align}
They satisfy not only the identities
\begin{align}
\begin{aligned}
 X^2 = Y^2 = Z^2 &= I, & XY &= iZ, &  YZ &=iX, & ZX &= iY
\end{aligned}
\end{align}
but they also form a basis for the space of $2\times2$ Hermitian
matrices.  We can expand any single qubit density matrix
$\varphi$ in terms of Pauli matrices using the expression
\begin{align}
\varphi = \frac{1}{2} \left( I + xX + yY + zZ \right).
\end{align}
If we collect the coefficients above, then
$(x, y, z) \in \mathbb{R}^3$ is the
\emph{Bloch vector} of $\varphi$.

A $n$-qubit stabilizer circuit is limited to certain quantum
gates and measurements.  It may use elements from the Clifford
group $\mathcal{C}(n)$, and it may apply measurements in the
$Z$-basis.  The Clifford group is generated by the
Controlled-NOT ($\cnot$), Hadamard ($H$), and
Phase ($P$) operators:
\begin{align}
\cnot & = \cnotgatematrix, &
H & = \hgatematrix,  &
P & = \sgatematrix.
\end{align}
A stabilizer circuit thus contains entirely of $\cnot$, $H$,
and $P$ gates.  For the values of $n$ we are concerned with,
$\mathcal{C}(1)$ and $\mathcal{C}(2)$ have sizes $24$ and
$11520$, respectively, modulo global phases. The circuit
diagram for a $Z$-measurement is given by the left image below:
\begin{align*}
  & \Qcircuit @C=1em @R=0.2em {
    & \measureD{Z} &
  }
  & \Qcircuit @C=1em @R=1.2em {
    & \qswap      & \qw \\
    & \qswap \qwx & \qw
  } & \end{align*}
while the right image represents a qubit $\swap$.
A Clifford circuit is a stabilizer circuit that excludes
measurements and implements a Clifford group unitary only.

\section{Postselected Two-to-One Stabilizer Circuits}
We revisit the study of stabilizer reductions from \cite{reichardt.msd.2006}
to derive Lemma \ref{lemma:cir-outputs-all}. Part of the novelty that
Lemma~\ref{lemma:cir-outputs-all} brings is the realization of recovery
circuits described in the next section. We first introduce some terminology
and notation to more concisely capture Reichardt's observations in
\cite{reichardt.msd.2006} to present our result.

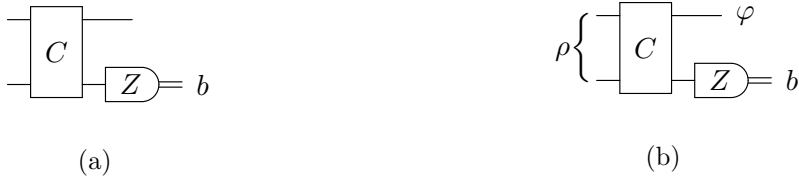
\begin{figure}[t]
  \centering
  \begin{subfigure}[h]{0.45\textwidth}
    \begin{align*}
    \Qcircuit @C=0.8em @R=1.2em {
      \lstick{} & \multigate{1}{C} &
         \qw \\
      \lstick{} & \ghost{C} &
         \measureD{Z} & \rstick{b} \cw
    }
    \end{align*}
    \caption{}
    \label{fig:2to1-cir-ex}
  \end{subfigure}
  \begin{subfigure}[h]{0.45\textwidth}
    \begin{align*}
    \rho \bigg\{\,\begin{matrix}
    \Qcircuit @C=0.8em @R=1.2em {
      \lstick{} & \multigate{1}{C} &
         \rstick{\varphi} \qw \\
      \lstick{} & \ghost{C} &
         \measureD{Z} & \rstick{b} \cw
    } \end{matrix}
    \end{align*}
    \caption{}
    \label{fig:2to1-cir-output}
  \end{subfigure}
  \caption{
    (a) A postselected two-to-one stabilizer circuit $(C, b)$
    consists of a stabilizer circuit component $C$ and a postselected
    bit value $b$.
    (b) The qubit $\varphi = \Phi_b(C,\rho)$ is the output of
    a postselected two-to-one stabilizer circuit $(C, b)$ on the
    two-qubit input $\rho$.
  }
  \label{fig:2to1ex}
\end{figure}

A $n$-to-$1$ stabilizer reduction is a procedure that accepts a $n$-qubit
state and generates a single qubit output using stabilizer operations only.
This means all post-measurement activities are also restricted to classical
control over stabilizer operations. Reichardt showed that any reduction can
be standardized to a particular form: an application of a Clifford unitary
on $n$ qubits, followed by a projection of qubits $2$ to $n$ onto a
computational basis state \cite{reichardt.msd.2006}. Since our focus is on
$n = 2$, we have the following definition.

\begin{definition}[postselected two-to-one stabilizer circuit]
A \emph{postselected two-to-one stabilizer circuit}
$(C, b)$ is a two-qubit quantum circuit that implements
a Clifford unitary $C$, followed by a $Z$-measurement
on the second qubit with an outcome $b \in \{0, 1\}$.
\end{definition}

\begin{definition}[probability and output]
Let $(C, b)$ be a postselected two-to-one stabilizer circuit and
let $\rho$ be a two-qubit state.  Then the \emph{probability} $Q_b$
of outcome $b$ on the transformed state $C\rho C^{\dagger}$ is
\begin{align}
  Q_b(C, \rho) = \mathrm{Tr}((I \otimes \bra{b}) C \rho
                               C^\dagger (I\otimes\ket{b})).
\end{align}
If $Q_b(C, \rho) > 0$,
then the \emph{output} $\Phi_b$ of a postselected circuit
$(C, b)$ on an input $\rho$ is
\begin{align}
\Phi_b(C, \rho) =
  \frac{(I \otimes \bra{b}) C \rho C^\dagger (I\otimes\ket{b})}
       {Q_b(C, \rho)}.
\end{align}
\end{definition}

At times, we may say \emph{circuit $C$} to reference the stabilizer
circuit piece only of the postselected circuit, which includes the
measurement at the end.  This allows us to use language about how
we may use or run the circuit $C$ and is often followed by details
on what course of action to take conditional on $b$ (or $1-b$).
The next definition describes what it means for postselected circuits
to produce similar outputs.

\begin{definition}[equivalent postselected two-to-one
stabilizer circuits] \label{def:equiv-cir}
Two postselected two-to-one stabilizer circuits $(C_1,b_1)$ and
$(C_2,b_2)$ are \emph{Clifford equivalent}, $(C_1, b_1) \sim (C_2, b_2)$,
if and only if there is a single qubit Clifford gate $G$ such that
for all two-qubit states $\rho$, we have the equality
\begin{align}
\label{eq:cliff-eq-cir}
(I \otimes \bra{b_1}) C_1 \rho C_1^\dagger (I\otimes\ket{b_1}) & =
G (I\otimes \bra{b_2}) C_2 \rho C_2^\dagger (I\otimes\ket{b_2}) G^\dagger.
\end{align}
Note that a Clifford equivalence implies that the probabilities
of observing a $b_1$ or $b_2$ are the same for the two circuits
i.e. $Q_{b_1}(C_1, \rho) = Q_{b_2}(C_2, \rho)$.
We say two postselected circuits are simply \emph{equivalent},
$(C_1, b_1)\equiv (C_2, b_2)$, if and only if $G = I$ in
Equation~\ref{eq:cliff-eq-cir}.
\end{definition}

We may also alter the circuits using $\ket{b_2} = X\ket{1-b_2}$
in Equation~\ref{eq:cliff-eq-cir} so that both postselect on the same
value.

As we mentioned earlier, any two-to-one stabilizer reduction can be
achieved through a postselected two-to-one stabilizer circuit. Despite
$|\mathcal{C}(2)| = 11520$, the number of actual reductions we need to
consider is $30$: one for each nontrivial two-qubit Pauli, plus the bit
\cite{reichardt.msd.2006}. As such, we can introduce three forms in the
following lemma to represent all postselected circuits $(C, b)$.
The proof is provided in Appendix~\ref{app:proof-eq-cir}.

\begin{figure}[t]
  \begin{subfigure}[h]{0.31\textwidth}
    \begin{align*}
        \Qcircuit @C=1.0em @R=1.2em {
          \lstick{} & \gate{G_1} & \ctrl{1} &
              \gate{G_3} & \qw \\
          \lstick{} & \gate{G_2} & \targ &
              \measureD{Z} & \rstick{0} \cw
        }
    \end{align*}
    \caption{}
    \label{fig:2to1-cir-cls-03}
  \end{subfigure}
  \begin{subfigure}[h]{0.31\textwidth}
    \begin{align*}
        \Qcircuit @C=1.0em @R=1.2em {
          \lstick{} & \gate{G_1} &
              \qw \\
          \lstick{} & \gate{G_2} &
             \measureD{Z} & \rstick{0} \cw
        }
    \end{align*}
    \caption{}
    \label{fig:2to1-cir-cls-01}
  \end{subfigure}
  \begin{subfigure}[h]{0.31\textwidth}
    \begin{align*}
        \Qcircuit @C=1.0em @R=1.2em {
          \lstick{} & \qswap & \gate{G_1} &
              \qw \\
          \lstick{} & \qswap \qwx & \gate{G_2} &
             \measureD{Z} & \rstick{0} \cw
        }
    \end{align*}
    \caption{}
    \label{fig:2to1-cir-cls-02}
  \end{subfigure}
  \caption{
    Any stabilizer procedure generating one qubit from two can be
    described by a postselected circuit $(C, b)$ resembling circuit
    (a), circuit (b), or circuit (c).
    The choice of single qubit Clifford gates $G_1$, $G_2$,
    and $G_3$ depend on $C$ and the measurement $b$.
    Circuit (a) is known as an interacting postselected
    circuit; the precise definition is provided in
    Section~\ref{sec:rec-cir}.
  }
  \label{fig:2to1-cir-cls}
\end{figure}
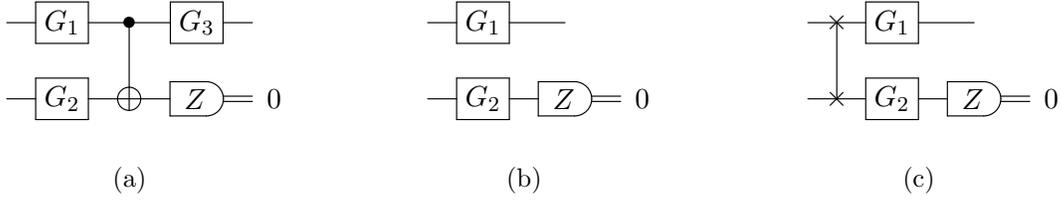

\begin{lemma}
\label{lemma:cir-outputs-all}
For every postselected two-to-one stabilizer circuit $(C, b)$,
there exists single qubit Clifford gates
$G_1$ and $G_2$ such that
either $(C, b) \sim (I \otimes G_1, 0)$,
or $(C, b) \sim ((I \otimes G_1) \swap, 0)$,
or $(C, b) \sim (\cnot(G_1 \otimes G_2), 0)$.
\end{lemma}

\begin{corollary}
\label{cor:cir-negate-eq}
Let $(C_{eq}, b_{eq}) \sim (C, b)$.
Then $(C,1-b)$ is Clifford equivalent to a slightly modified
version $(C'_{eq},b'_{eq})$ of $(C_{eq},b_{eq})$:
\begin{align}
\begin{array}{c|c}
(C_{eq},b_{eq})\sim (C,b) & (C'_{eq},b'_{eq})\sim (C,1-b)\\\hline
(I \otimes G_1, 0) & (I \otimes XG_1, 0)\\
((I \otimes G_1) \swap, 0) & ((I \otimes XG_1) \swap, 0)\\
(\cnot(G_1 \otimes G_2), 0) & ((I \otimes X)\cnot(G_1 \otimes G_2), 0)\\
\end{array} &
\end{align}
\end{corollary}

Due to Lemma~\ref{lemma:cir-outputs-all}, we have a remarkably
much easier time studying postselected circuits.  We may substitute
$(C, b)$ with another that likely uses fewer gates but behaves
in exactly the same way.  Because there are many identities on
Pauli operators and Clifford gates, $G_1$ and $G_2$
are not unique e.g.\
$((\cnot(Z \otimes I), 0)\equiv ((Z \otimes I)\cnot, 0)
\sim (\cnot, 0)$.  Of the $30$ reductions available, it is easy to
see that there are $18$ varieties of $(\cnot(G_1 \otimes G_2), 0)$,
and $6$ each for $(I \otimes G_1, 0)$ and $((I \otimes G_1)\swap, 0)$.
If we want to separate the circuits by the stricter kind of equivalence,
the number of classes is multiplied by $24$ e.g.\ $18 \cdot 24 = 432$
for $((G_3 \otimes I)\cnot(G_1 \otimes G_2), 0)$,
since there are $|\mathcal{C}(1)| = 24$ choices of $G_3$.

\section{Recovery Circuits}
\label{sec:rec-cir}
A quantum circuit involving measurements likely has outcomes
that we prefer over others. If we are less than fortunate,
convention dictates that we discard the output and rerun the
circuit on new input instances until we succeed. This is not
much of an issue when the initial overhead is low, but can
become problematic otherwise. If the cost associated with
state preparation is a barrier to large computations, any
method that alleviates this burden is highly desirable. It
turns out when $\rho$ is a tensor product state,
i.e.\ $\rho = \varphi \otimes \ketbra{\psi}{\psi}$, we have
an alternative: there exist operations capable of reusing
an undesirable output to try and recovery $\varphi$.

This input requirement means the only circuit
configuration of Lemma~\ref{lemma:cir-outputs-all}
worth considering is $(\cnot(G_1 \otimes G_2), 0)$.
We can easily see that when $(C, b) \sim (I \otimes G_1, 0)$,
the output of $(C, b)$ on $\varphi_1 \otimes \varphi_2$
is essentially $\varphi_1$.  The output is always an
input, and the same is similarly true for all circuits
$(C, b) \sim ((I \otimes G_1)\swap, 0)$.

\begin{definition}[interacting postselected circuit]
A postselected two-to-one stabilizer circuit $(C, b)$ is
\emph{interacting} if and only if there
are single qubit Clifford gates $G_1$ and $G_2$ such that
$(C, b) \sim (\cnot(G_1 \otimes G_2), 0)$.
We say circuit $C$ is \emph{interacting} if and only if
$(C, 0)$ is interacting.
\end{definition}

With that, we define the notion of a recovery circuit. As a
matter of convenience, we use $\psi$ in place $\ketbra{\psi}{\psi}$
throughout the remainder of our discussion on recovery circuits.

\begin{definition}[recovery circuit]
Let $(C, b)$ be an interacting postselected circuit.
Then a postselected two-to-one stabilizer circuit $(C', b')$
is a \emph{recovery circuit} of $(C, b)$ if and only if
for all two-qubit states $\varphi \otimes \psi$,
we have
\begin{align}
\varphi = \Phi_{b'}\left(
              C', \Phi_{1-b}(C, \varphi \otimes \psi)
              \otimes \psi \right).
\end{align}
\end{definition}

Notice that an input qubit to $(C', b')$ is
$\Phi_{1-b}(C, \varphi \otimes \psi)$, the output from
the circuit that postselects on the opposite outcome.
In this context, if $b$ is more desirable than $1 - b$, then
we say circuit $C$ is \emph{successful} upon measuring $b$ on
$C\left( \varphi \otimes \psi \right) C^{\dagger}$.
Otherwise circuit $C$ is \emph{unsuccessful}, and the recovery
circuit provides a second chance at obtaining the output of
$(C, b)$ on $\varphi \otimes \psi$ by using a presumably much
simpler circuit $C'$ to recover $\varphi$. Our next lemma presents
one way on how to design such a recovery circuit to $(C, b)$.

\begin{lemma}
\label{lemma:rec-cir}
Every interacting postselected circuit $(C, b)$
has a recovery circuit.
\end{lemma}
\begin{proof}
Let $(C, b) \sim (\cnot(G_1 \otimes G), 0)$, where
$G_1$ and $G$ are single qubit Clifford gates.
By Corollary~\ref{cor:cir-negate-eq}, we know
\begin{align}
(C, 1-b) \sim ((I \otimes X)\cnot(G_1 \otimes G), 0)
         \equiv (\cnot(G_1 \otimes G), 1)
\end{align}
which means there is a single qubit Clifford
gate $G_2$ such that
\begin{align}
(C, 1-b) \equiv ((G^{\dagger}_2 \otimes I)\cnot(G_1 \otimes G), 1).
\end{align}
We shall show that $((G^{\dagger}_1 \otimes I)
\cnot(G_2 \otimes G), 0)$
is a recovery circuit of $(C, b)$.  Figure~\ref{fig:democir}
includes reference diagrams to aid comprehension.

If the input to circuit $C$ is $\varphi_1 \otimes
\psi$, consider $\varphi'_1 \otimes \psi' =
G_1\varphi_1 G^{\dagger}_1 \otimes G\psi G^{\dagger}$.
Let $(x_1, y_1, z_1)$ be the Bloch vector of $\varphi'_1$
and $(x, y, z)$ be the Bloch vector of $\ket{\psi'}$. For
ease of notation, we define outputs
\begin{align}
\varphi'_2 &= \Phi_1(\cnot, \varphi_1' \otimes \psi') \\
\varphi_2 &= G^{\dagger}_2 \varphi'_2 G_2
           = \Phi_{1-b}(C, \varphi_1 \otimes \psi).
\end{align}
Then the Bloch vector $(x_2, y_2, z_2)$
of $\varphi'_2$ becomes
\begin{align}
  x_2 &= \frac{x_1x + y_1y}
               {1 - z_1z}, &
  y_2 &= \frac{y_1x - x_1y}
               {1 - z_1z}, &
  z_2 &= \frac{z_1 - z}
               {1 - z_1z}.
\end{align}

Now suppose $\varphi_3 = \Phi_0(\cnot(G_2 \otimes G),
\varphi_2 \otimes \psi)$.
For postselected circuits that are basically a single $\cnot$,
the equations for computing the output's Bloch vector
are essentially the same:
\begin{align}
  x_3 &= \frac{x_2x - y_2y}
               {1 + z_2z}, &
  y_3 &= \frac{y_2x + x_2y}
               {1 + z_2z}, &
  z_3 &= \frac{z_2 + z}
               {1 + z_2z},
\end{align}
where $(x_3, y_3, z_3)$ represents the Bloch vector of
$\varphi_3$.  Using $x^2 + y^2 + z^2 = 1$, we can show
\begin{align}
  x_3 &=
  \frac{x_1x^2 + xy_1y - xy_1y + x_1y^2}
       {1 - z_1z + z_1z - z^2} = x_1.
\end{align}
Likewise, $y_3 = y_1$ and $z_3 = z_1$, which means
$\varphi_3 = \varphi_1' = G_1\varphi_1 G^{\dagger}_1$.
The circuit
$((G^{\dagger}_1 \otimes I)\cnot(G_2 \otimes G), 0)$
is therefore a recovery circuit of $(C, b)$.
\end{proof}

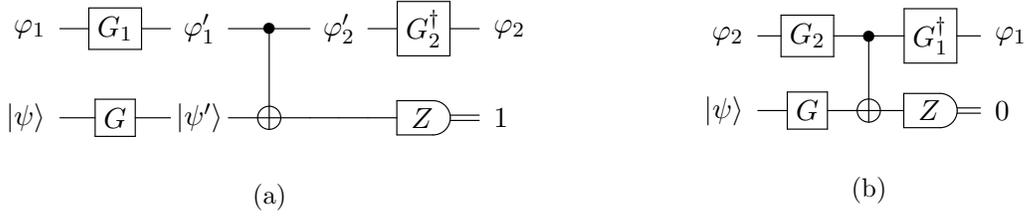
\begin{figure}[t]
  \begin{subfigure}[h]{0.55\textwidth}
    \begin{align*}
      \Qcircuit @C=1.0em @R=1.5em {
        \lstick{\varphi_1} & \gate{G_1} & \qw & \varphi'_1 & &
            \ctrl{1} & \qw & \varphi_2' & &
            \gate{G^{\dagger}_2} & \rstick{\varphi_2}
            \qw \\
        \lstick{\ket{\psi}} & \gate{G} & \qw & \ket{\psi'} & &
            \targ    & \qw & \qw        & \qw &
            \measureD{Z} & \rstick{1} \cw
      }
    \end{align*}
    \caption{}
    \label{fig:democir2}
  \end{subfigure}
  \begin{subfigure}[h]{0.4\textwidth}
    \begin{align*}
      \Qcircuit @C=0.8em @R=1.0em {
        \lstick{\varphi_2} & \gate{G_2} &
            \ctrl{1} &
            \gate{G^{\dagger}_1} & \rstick{\varphi_1} \qw \\
        \lstick{\ket{\psi}} & \gate{G} &
            \targ &
            \measureD{Z} & \rstick{0} \cw
      }
    \end{align*}
    \caption{}
    \label{fig:democir3}
  \end{subfigure}
  \caption{Suppose $(C, 1-b)\equiv ((G^{\dagger}_2 \otimes I)
           \cnot(G_1 \otimes G), 1)$.
           This equivalence allows us to study $(C, 1-b)$
           via its substitute in (a) and come up with
           the recovery circuit in (b). We include
           intermediate states like  $\varphi_1'$ and
           $\varphi_2' = G_2 \varphi_2 G^{\dagger}_2$
           in (a) to signify stages in the circuit.
  }
  \label{fig:democir}
\end{figure}

Between $(C, b)$ and its recovery circuit
$((G^{\dagger}_1 \otimes I)\cnot(G_2 \otimes G), 0)$,
there is a relatively straightforward relationship
between the probability that circuit $C$ would have
been successful and the probability that circuit
$(G^{\dagger}_1 \otimes I)\cnot(G_2 \otimes G)$
will be successful.

\begin{corollary}
\label{cor:rec-prob}
Let $\varphi_1 \otimes \psi$ be a two-qubit state and
let $C' = (G^{\dagger}_1 \otimes I)\cnot(G_2 \otimes G)$
be a two-qubit Clifford unitary such that $(C', 0)$
is a recovery circuit of $(C, b)$.
Then
\begin{align}
Q_0(C', \Phi_{1-b}(C, \varphi_1 \otimes \psi)
        \otimes \psi)
    & = \frac{(1 - z^2)/4}
             {1 - Q_b(C, \varphi_1 \otimes \psi)}
    \label{eq:Q1Q2}
\end{align}
where $z = \bra{\psi}G^{\dagger}ZG\ket{\psi}$.
\end{corollary}
\begin{proof}
We assume for simplicity that $C = \cnot$
and $b = 0$, which implies $G_1 = G_2 = G = I$.
Let $z_1 = \mathrm{Tr}(Z\varphi_1)$ and
$z = \bra{\psi}Z\ket{\psi}$.  Also let
$\varphi_2 = \Phi_1(C, \varphi_1 \otimes \psi)$.
Then
\begin{align}
  Q_1(C, \varphi_1 \otimes \psi)
  = \frac{1 - z_1z}{2}.
\end{align}
Similarly, $z_2 = \mathrm{Tr}(Z\varphi_2) =
\frac{z_1 - z}{1 - z_1z}$.  The probability of
recovering $\varphi_1$ is now clear:
\begin{align}
  Q_0(C', \varphi_2 \otimes \psi)
  = \frac{1 + z_2z}{2}
  &= \frac{1 - z_1z + z_1z - z^2}
          {4 \left(\frac{1 - z_1z}{2} \right)} \\
  &= \frac{(1 - z^2)/4}
          {1 - Q_0(C, \varphi_1 \otimes \psi)}
\end{align}
since the circuits perform a single measurement.
\end{proof}

Another implication of the proof to Lemma~\ref{lemma:rec-cir}
is that $\Phi_{1-b}(C, \varphi_1 \otimes \psi)$ is always
$\varphi_1$, up to a single qubit Clifford gate, whenever
$\ket{\psi}$ is an eigenstate of $X$, $Y$, or $Z$ (a stabilizer
qubit). Under these circumstances, the behavior of $(C, b)$ on
these types of inputs is actually no different than non-interacting
circuits.  Hence it does not warrant the use of a circuit
$(G^{\dagger}_1 \otimes I)\cnot(G_2 \otimes G)$ to try and
perform a recovery because the qubit is basically $\varphi_1$.
It is also quite evident by now that there is only one type of
recovery circuit, especially given our construction in Lemma
\ref{lemma:rec-cir}.

\begin{lemma}
\label{lemma:rec-cir-interact}
All recovery circuits are interacting postselected circuits.
\end{lemma}
\begin{proof}
Let $(C, b)$ be an interacting postselected circuit and suppose
$(C', b')$ is a recovery circuit of $(C, b)$. If $(C', b')$ is
not an interacting postselected circuit, then
$(C', b') \sim (I \otimes G, 0)$ or
$(C', b') \sim ((I \otimes G)\swap, 0)$,
where $G$ is a single qubit Clifford gate. We can easily find a
two-qubit state $\varphi \otimes \psi$ such that $(C', b')$ fails
to recover $\varphi$ on the input
$\Phi_{1-b}(C, \varphi \otimes \psi) \otimes \psi$.
\end{proof}

Lastly, it should not come as a surprise that more than
one recovery circuit of $(C, b)$ exists.
Even so, we can guarantee that not any one recovery
circuit will outperform another.

\begin{lemma}
\label{lemma:rec-cir-equal}
Let $(C, b)$ be an interacting postselected circuit, and let
$C'' = (G^{\dagger}_2 \otimes I)\cnot(G_1 \otimes G)$ be a
two-qubit Clifford unitary such that $(C, 1-b)\equiv (C'', 1-b'')$.
Then $(C', b')$ is a recovery circuit of $(C, b)$ if and
only if $(C', b')\equiv ((G^{\dagger}_1 \otimes I)\cnot
(G_2 \otimes G), b'')$.
\end{lemma}
\begin{proof}
In the reverse direction, equivalence of postselected
circuits means both produce the exact same output
at the same success rate for all two-qubit states $\rho$.
This certainly includes all two-qubit product states
$\varphi_2 \otimes \psi$,
where $\varphi_2$ is the output of $(C, 1-b)$ on
another input $\varphi_1 \otimes \psi$.

In the forward direction, Lemmas~\ref{lemma:clifford-maps}
and \ref{lemma:rec-cir-equiv} in the Appendices do most of
the job: $(C', b') \sim ((G^{\dagger}_1 \otimes I)\cnot
(G_2 \otimes G), b'')$.  We just need to prove equivalence.
We look back at the definition of Clifford equivalent
postselected circuits, where we must have a single qubit
Clifford gate $G$ such that
\begin{align}
(I \otimes \bra{b'})C'\rho C'^{\dagger}(I \otimes \ket{b'}) =
(G \otimes \bra{b''})C''\rho C''^{\dagger}(G^{\dagger} \otimes \ket{b''})
\end{align}
for all two-qubit states $\rho$.  If it is indeed the case
that they are strictly Clifford equivalent i.e. $G \ne I$,
then $(C', b')$ cannot have been a recovery circuit of $(C, b)$
because the output from $(C', b')$ on $\rho$ will be rotated by $G$.
Thus the two must be equivalent (with ``$\equiv$'').
\end{proof}

From Lemmas~\ref{lemma:rec-cir} and \ref{lemma:rec-cir-equal},
we reach our main result, with Corollary~\ref{cor:rec-cir-rec}
as an immediate consequence to our theorem.

\begin{theorem}
\label{thm:rec-cir}
Every interacting postselected circuit $(C, b)$
has a recovery circuit $(C', b')$.
Moreover, all recovery circuits of $(C, b)$ are
equivalent to $(C', b')$.
\end{theorem}

\begin{corollary}
\label{cor:rec-cir-rec}
Every recovery circuit $(C', b')$ has its own recovery
circuit $(C'', b'')$ with $b'' = b'$.
\end{corollary}
\begin{proof}
The first claim is a given since recovery circuits
are interacting postselected circuits themselves by
Lemma~\ref{lemma:rec-cir-interact}. If the recovery
circuit $(C'', b'')$ already satisfies $b'' = b'$,
then we are done. Otherwise, $b'' = 1 - b'$,
so $(C'', b'') \equiv ((I \otimes X)C'', b')$.
\end{proof}

\section{Example Routines Featuring Recovery Circuits}

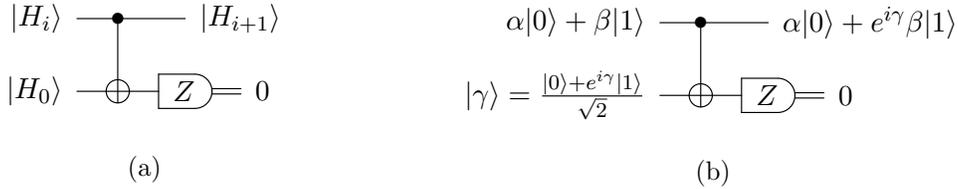
\begin{figure}[t]
  \begin{subfigure}[h]{0.45\textwidth}
    \begin{align*}
      \begin{matrix} \quad
      \Qcircuit @C=1.0em @R=1.8em {
        \lstick{\ket{H_i}} & \ctrl{1} & \rstick{\ket{H_{i+1}}}
            \qw \\
        \lstick{\ket{H_0}} &  \targ   & \measureD{Z} & \rstick{0} \cw
      }
      \end{matrix}
    \end{align*}
    \caption{}
    \label{fig:ex-dcs}
  \end{subfigure}
  \begin{subfigure}[h]{0.45\textwidth}
    \begin{align*} \quad \quad
      \Qcircuit @C=1.0em @R=1.8em {
        \lstick{\alpha\ket{0} + \beta\ket{1}} & \ctrl{1} &
            \rstick{ \alpha\ket{0} +
                     e^{i\gamma}\beta\ket{1}} \qw \\
        \lstick{\ket{\gamma} =
                \frac{\ket{0} + e^{i\gamma}\ket{1}}{\sqrt{2}}} &
            \targ & \measureD{Z} & \rstick{0} \cw
      }
    \end{align*}
    \caption{}
    \label{fig:ex-par}
  \end{subfigure}
  \caption{The postselected circuit $(\cnot, 0)$ appears
      in both \cite{dcs} and \cite{jones.qchem}, with each
      one supplying a different input set to circuit $\cnot$.
      The qubit $\ket{H_0}$ in (a) is the $+1$ eigenstate of
      $H$, and the process of generating $\ket{H_{i+1}}$ in
      \cite{dcs} starts with $\ket{H_0} \otimes \ket{H_0}$.
      For (b), a qubit $\ket{\gamma}$ leads to a $+\gamma$
      $Z$-rotation on the first qubit $\alpha\ket{0}+\beta\ket{1}$
      upon measuring $0$.
  }
  \label{fig:excir}
\end{figure}

Recovery circuits appear in the literature. The procedure
of \cite{dcs} uses an interacting circuit with a single
$\cnot$ to obtain ``ladder'' qubit states of the kind
\begin{align}
\ket{H_i} =
\cos\left(\theta_i\right)\ket{0} + \sin\left(\theta_i\right)\ket{1},
    \ \cot\left(\theta_i\right) = \cot^{i+1}\left(\pi/8\right)
\end{align}
for $i \ge 0$. The production starts by supplying two copies
of the magic state $\ket{H_0} = H\ket{H_0}$ to the $\cnot$
circuit in Figure~\ref{fig:ex-dcs}.  Each time we gain a new
state $\ket{H_i}$, we can reuse the qubit in an attempt to
create the next $\ket{H_{i+1}}$. If the attempt fails, then
the output of $(\cnot, 1)$ on $\ket{H_i} \otimes \ket{H_0}$
is $\ket{H_{i-1}}$. Given that the recovery circuit of
$(\cnot, 0)$ is itself, the method to recover $\ket{H_i}$
from $\ket{H_{i-1}} \otimes \ket{H_0}$ is no different than
the procedure to create it.

The programmable ancilla rotation (PAR) of \cite{jones.qchem}
is similar.  It uses qubits of the kind
\begin{align}
\ket{\gamma} = \frac{\ket{0} + e^{i\gamma}\ket{1}}{\sqrt{2}}
\label{eq:equa-qbit}
\end{align}
to rotate $\ket{q} = \alpha\ket{0} + \beta\ket{1}$
about the $Z$-axis by an angle $\gamma$, demonstrated in
Figure~\ref{fig:ex-par}.  The PAR routine uses
the same interacting circuit $\cnot$ to achieve this task.
On the chance that the $Z$-measurement returns $1$,
then instead of
$\ket{q + \gamma} = \alpha\ket{0} + e^{i\gamma}\beta\ket{1}$,
the output becomes
$\ket{q - \gamma} = \alpha\ket{0} + e^{-i\gamma}\beta\ket{1}$,
which is $\ket{q}$ rotated by $-\gamma$.  Jones et.\ al
\cite{jones.qchem} suggest pairing $\ket{q - \gamma}$
with $\ket{2\gamma}$ as a direct line to $\ket{q + \gamma}$,
but we can alternatively break this down into two smaller
steps if $\ket{\gamma}$ are the only states available.  We
first run the $\cnot$ circuit on
$\ket{q - \gamma} \otimes \ket{\gamma}$.
If we measure $0$, then we recover $\ket{q}$, and
we proceed with rerunning circuit $\cnot$ on
$\ket{q} \otimes \ket{\gamma}$.

\begin{figure}[t]
  \begin{subfigure}[b]{0.55\textwidth}
    \begin{align*}
    \Qcircuit @C=1.2em @R=1.0em {
                                             &                       &
        \multigate{4}{\ \ U\ \ }  & \qw &  \qw           &
        \qw               &  \qw           & \\
                                             &                       &
        \ghost{\ \ U\ \ }         & \qw &  \qw           &
        \qw               &  \qw           & \\
      \lstick{\raisebox{1.2em}{ $\rho$}~~ }  &                       &
        \ghost{\ \ U\ \ }         & \qw &  \qw           &
        \qw               &  \qw           & \\
                                             &  \vdots               &
                                  &     &                &
                          &  \vdots        & \\
                                             &                       &
        \ghost{\ \ U\ \ }         & \qw &  \varphi~~~~   &
        \multigate{1}{C}  &  \qw  & \push{\rule{0em}{0.1em}} \\
                                             &  \lstick{\ket{\psi}}  &
        \qw                       & \qw &  \qw           &
        \ghost{C}         &  \measureD{Z} & \rstick{b} \cw
    \gategroup{1}{1}{5}{1}{0.5em}{\{}
    }
    \end{align*}
    \caption{}
    \label{fig:ex-entangle}
  \end{subfigure}
  \begin{subfigure}[b]{0.40\textwidth}
    \begin{align*}
    \Qcircuit @C=0.8em @R=0.8em {
      \lstick{\varphi} &
         \multigate{1}{C} &
         \rstick{\varphi'} \qw \\
      \lstick{\ket{\psi}} & \ghost{C} &
         \measureD{Z} & \rstick{1-b} \cw
    }
    \end{align*}
    \begin{align*}
    \Qcircuit @C=0.8em @R=0.8em {
      \lstick{\varphi'} &
         \multigate{1}{C'} &
         \rstick{\varphi} \qw \\
      \lstick{\ket{\psi}} & \ghost{C'} &
         \measureD{Z} & \rstick{b'} \cw
    }
    \end{align*}
    \vspace{1em}
    \caption{}
    \label{fig:ex-entangle-rec}
  \end{subfigure}
  \caption{Recovery circuits also work when one of the
    qubits is entangled with another system.
    In (a), we trace out all but the $n$-th qubit of
    $U\rho U^{\dagger}$ to get $\varphi \otimes \psi$ as input
    to circuit $C$.  If we measure $1-b$ as pictured in the
    top circuit of (b), then we execute circuit $C'$ on
    $\varphi' \otimes \psi$ to try and recover $\varphi$.
    We succeed with the recovery if we measure $b'$.
  }
  \label{fig:larger-rec}
\end{figure}
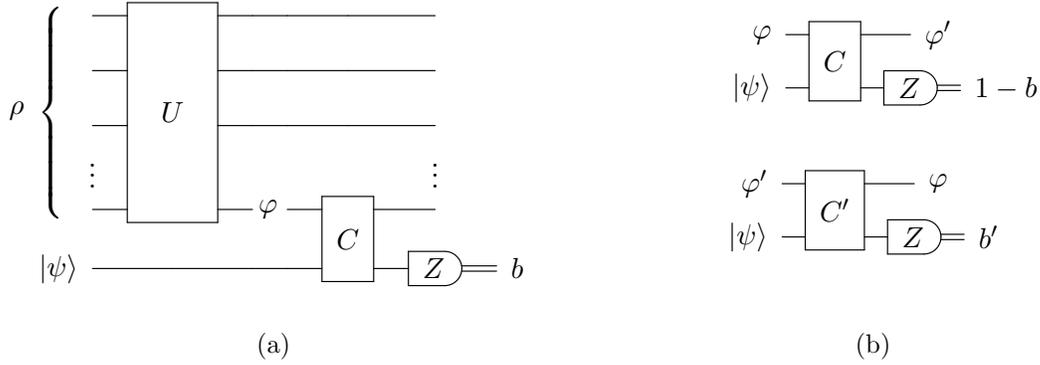

In addition to pure qubits, our notation for the two-qubit
state $\varphi \otimes \psi$ throughout the
previous section indicates that $\varphi$ is allowed to be
mixed, and it can even be part of a larger entangled system.
As a quick demonstration, suppose we have the situation as
illustrated in Figure~\ref{fig:ex-entangle}. Let $(C', b')$
be the recovery circuit of $(C, b)$ and let
\begin{align}
U \rho U^{\dagger}
&= \frac{1}{2^n} \left( \mathrm{\mathbf{P}}_I \otimes I +
                        \mathrm{\mathbf{P}}_X \otimes X +
                        \mathrm{\mathbf{P}}_Y \otimes Y +
                        \mathrm{\mathbf{P}}_Z \otimes Z  \right)
\end{align}
where $\mathrm{\mathbf{P}}_L$ are Pauli operator sums on the
first $n-1$ qubits.  While the proof to Lemma~\ref{lemma:rec-cir}
is generalizable to include the unused portions
$\mathrm{\mathbf{P}}_L$ of the entangled state, the math is
simpler and works out the same if we trace out the first $n-1$
qubits, keeping only the last qubit $\varphi = \mathrm{Tr}_{1,n-1}
\left( U\rho U^{\dagger} \right)$ that we need for the two-qubit
circuit.  If we are unlucky, then qubit $n$ becomes
$\varphi' = \Phi_{1-b}(C, \varphi \otimes \psi)$, but
we can try to regain $\varphi$ by executing circuit $C'$
on $\varphi' \otimes \psi$. If the recovery is successful,
then we have another opportunity at the output
$\Phi_b(C, \varphi \otimes \psi)$.
In all likelihood, this is a less lengthy process than
preparing another $\rho$ and running the circuit of $U$ again;
by some estimates, a synthesis of $U$ over a universal
gate set may require an exponential number of
gates \cite{giles.selinger.multi}.
This is a stark contrast to $C'$, which uses one
$\cnot$ with possibly a couple more single qubit
Clifford gates.

\section{Experimentation with Recovery Circuits}
\label{sec:expmts}
Consider a two-qubit Clifford unitary $C_1$ and a two-qubit
state $\varphi \otimes \psi$. Suppose we have a target outcome
of $b_1$; the intent is to produce output $\Phi_{b_1}(C_1,
\varphi \otimes \psi)$. Then by Corollary~\ref{cor:rec-cir-rec},
we can define a depth $k$ protocol to be a procedure on $k-1$
postselected circuits $(C_1, b_1)$, $\ldots$, $(C_{k-1}, b_{k-1})$
such that $(C_i, b_i)$ is the recovery circuit of $(C_{i-1}, b_{i-1})$.
We start by running circuit $C_1$ on $\varphi \otimes \psi$. If
circuit $C_1$ is successful i.e.\ we measure $b_1$, then no recovery
attempts are necessary and we declare success. Otherwise, we enlist
circuit $C_2$ to try and obtain $\varphi$. More generally, if
circuit $C_i$ is successful, then we recover an input qubit to circuit
$C_{i-1}$; if not, we run circuit $C_{i+1}$ to recover an input qubit
to circuit $C_i$.

The value of $k$ represents a stopping point in our protocol:
when circuit $C_{k-1}$ is unsuccessful, we declare failure,
discard the output, and restart with a new copy $\varphi \otimes
\psi$ to circuit $C_1$. Thus this process on $k-1$ circuits
is nothing more than a classical random walk on $k+1$ integers
$\{0, \ldots, k\}$, where the walk begins at location $1$, a
step onto $0$ signifies success, and a step onto $k$ means
failure. The success probability of circuit $C_i$ is the probability
of a left step from $i$ to $i-1$ and is determined recursively
by Equation~\ref{eq:Q1Q2} in Corollary~\ref{cor:rec-prob}. A
step in either direction consumes one $\ket{\psi}$.

We conduct simulations of this process to obtain a better idea
for $N_k$, the expected number of $\ket{\psi}$ resources needed
to create one $\Phi_{b_1}(C_1, \varphi \otimes \psi)$ with our
depth $k$ protocol. Let $d$ be the cost to prepare a single instance
of $\varphi$ relative to the cost of $\ket{\psi}$. Then the cost
of one execution or trial is the same as $d$ plus the number of
$\ket{\psi}$ qubits used before halting, regardless of declaring
success or fail. The costs from all trials are tallied together
and divided by the number of successful trials to obtain $N_k$.
We compare this against the expected cost without recovery
($k = 2$), which is
\begin{align}
N_2 = \frac{d + 1}{Q_{b_1}(C_1, \varphi \otimes \psi)}.
\end{align}
We assume for the sake of simplicity that $(C_1, b_1) = (\cnot, 0)$,
which means $(C_2, b_2) = (\cnot, 0)$, and so forth for the other
$k-3$ recovery circuits.

We further assume that $Q_0(CNOT, \varphi \otimes \psi) = 1/2$.
Since we fix the first success probability, $N_k$ is dependent
on the parameter $z = \bra{\psi}Z\ket{\psi}$ that appears in the
recovery success rate Equation~\ref{eq:Q1Q2}. Technically, we need
a different $\varphi$ with each choice of $\ket{\psi}$ to maintain
$Q_0(CNOT, \varphi \otimes \psi) = 1/2$ and the same output
$\Phi_0(CNOT, \varphi \otimes \psi)$. Usually different $\varphi$
means different costs $d$, but we will ignore this momentarily and
assume the preparation overhead $d$ for each $\varphi$ is the same
for the purposes of a broader comparison of $N_k$ across different
$\ket{\psi}$ qubits. In the first set of experiments, we include only
one recovery circuit ($k = 3$). The following table summarizes the
expected costs for four samples of $z$ obtained over the course of
$100000$ trials:
\begin{align*}
\begin{array}{|c|c|c|c|c|c|} \hline
  d & N_2 & N_3\colon z=\sqrt{0.96} & N_3\colon z=\sqrt{0.50} &
    N_3\colon z=\sqrt{0.04} & N_3\colon z=0 \\
  \hline
  10^{-1} &   2.2 &       3.20 &      3.18 &       3.15 &      3.15 \\
  \hline
     10^0 &   4.0 &       4.99 &      4.75 &       4.51 &      4.50 \\
  \hline
     10^1 &    22 &       22.7 &      20.5 &       18.2 &      18.0 \\
  \hline
     10^2 &   202 &      200.4 &     177.9 &      155.1 &     157.7 \\
  \hline
     10^3 &  2002 &     1988.9 &    1750.7 &     1521.9 &    1498.7 \\
  \hline
     10^4 & 20002 &    19816.4 &   17488.0 &    15215.4 &   14998.7 \\
  \hline
\end{array}
\end{align*}
The first row with $d = 0.1$ should be interpretted as $\varphi$
being cheaper to prepare than $\ket{\psi}$. We clearly see an
improvement when factoring in recovery in the face of large relative
preparation overhead between $\varphi$ and $\ket{\psi}$.  We also
see a trend of lower costs as $z$ grows smaller, when $\ket{\psi}$
is moving closer to the $XY$-plane in the Bloch sphere. This is due
to the differences in the recovery success rate at circuit $C_2$,
which are $0.02$, $0.25$, $0.48$, and $0.5$, respectively.

\begin{figure}[t]
  \centering
  \includegraphics[width=0.75\textwidth,trim={2em 1em 2em 2.5em},clip]
      {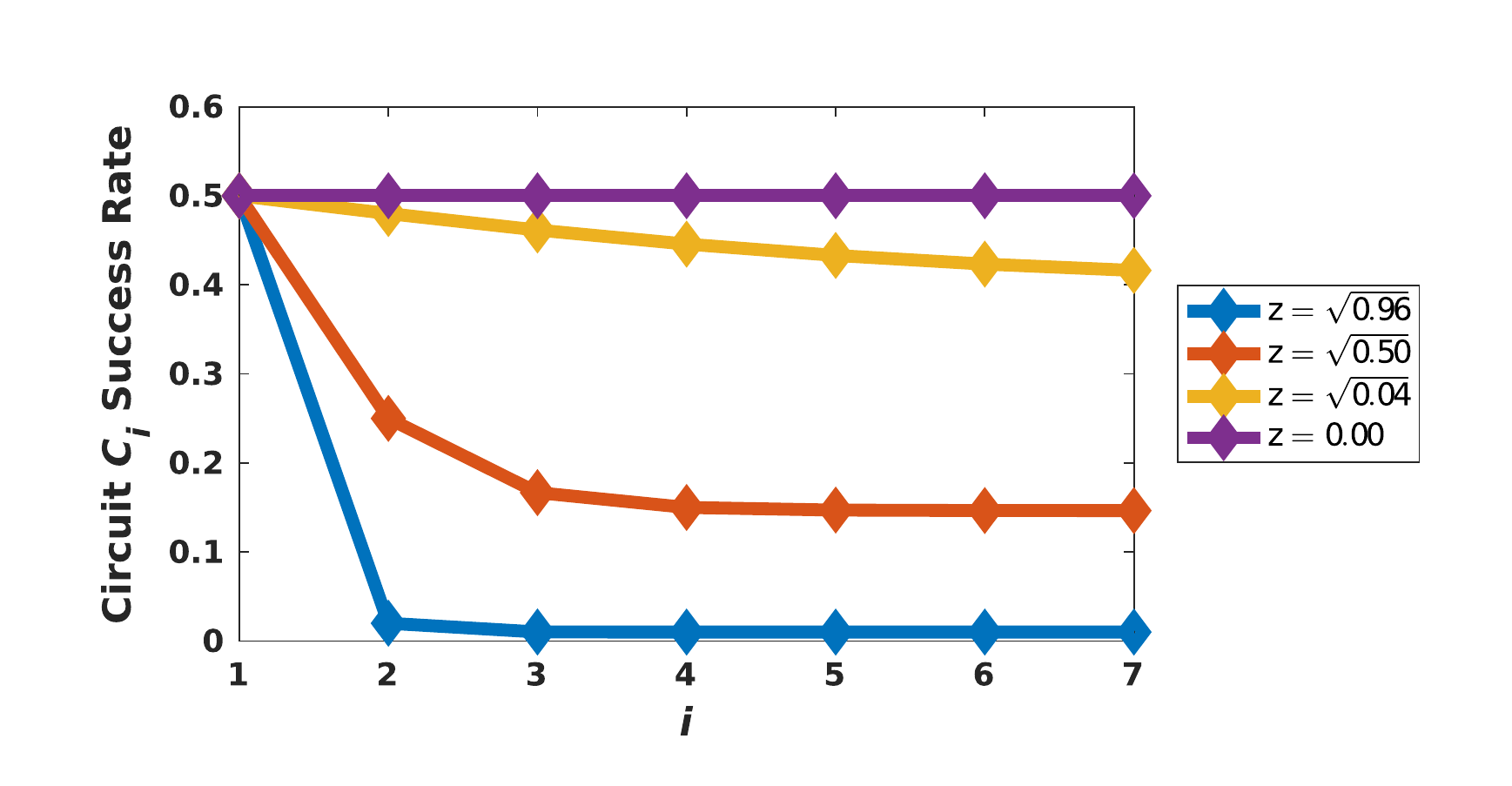}
  \caption{The success probability between circuit $C_i$ and
           circuit $C_{i+1}$ defined recursively in Corollary
           \ref{cor:rec-prob} drops more dramatically as $z$
           moves closer to $1$.  This leads to a greater expected
           cost $N_k$ of our protocol since the recovery is less
           likely to succeed relative to other choices of $z$.
           On the other end of the spectrum, the success probability
           of each circuit $C_i$ is uniform when $z=0$.}
  \label{fig:exp}
\end{figure}

In the second batch of experiments, we maintain $d = 1000$
but vary the number of circuits parameterized by $k$. Again,
$Q_0(CNOT, \varphi \otimes \psi) = 1/2$ and we run $100000$
trials. Data for $N_k$ is compiled together in the table below,
starting with $k = 3$:
\begin{align*}
\begin{array}{|c|c|c|c|c|} \hline
  k & N_k\colon z=\sqrt{0.96} & N_k\colon z=\sqrt{0.50} &
    N_k\colon z=\sqrt{0.04} & N_k\colon z=0 \\
  \hline
      3 &  1981.7 &  1753.2 &  1522.9 &  1501.6 \\
  \hline
      4 &  1982.9 &  1720.5 &  1372.2 &  1336.9 \\
  \hline
      5 &  1982.4 &  1716.5 &  1302.9 &  1255.2 \\
  \hline
      6 &  1987.5 &  1710.9 &  1266.6 &  1206.2 \\
  \hline
      7 &  1982.5 &  1715.3 &  1246.7 &  1174.7 \\
  \hline
     10 &  1991.7 &  1717.0 &  1221.5 &  1120.8 \\
  \hline
     20 &  2002.5 &  1727.3 &  1220.2 &  1072.9 \\
  \hline
     30 &  2006.3 &  1734.6 &  1231.4 &  1064.5 \\
  \hline
     40 &  2023.5 &  1743.7 &  1240.8 &  1066.3 \\
  \hline
\end{array}
\end{align*}
Observe that the value of $N_k$ continue to lower
noticably for some of $\ket{\psi}$ cases as more circuits
are added before increasing again. This behavior is no
surprise since at some point, the penalty to sustain the
recovery process will exceed the overhead of repeating the
computation. If we look at the success probabilities for
the first seven circuits of the protocol for each of the
four $z$ samples in Figure~\ref{fig:exp}, we also see the
success rates decrease to some lower boundary as $i$ increases,
with the exception of when $z = 0$. The drop in probabilities
from circuit $C_1$ to circuit $C_3$ is quite significant
when $z$ is close to $1$ (and $1 - z^2$ is small), so the
chance of recovery at circuit $C_3$ is only slightly larger
than $0$. This explains why there is no apparent change in
$N_k$ between one recovery circuit ($k = 3$) versus two ($k = 4)$
for the case $z = \sqrt{0.96}$. The ideal situation is
to know beforehand how many circuits to include to minimize
resource usage.

\section{Conclusion}
We have shown two-qubit stabilizer circuits require nothing
more than a few Clifford gates to perform a job.  These
simplifications shed light into the complementary nature
between interacting circuits.  Despite measurements generally
being irreversible, we find an exception when handling a
two-qubit product state $\varphi \otimes \psi$.
That is, we can use $\ket{\psi}$ in conjuction with a specific
circuit to salvage the expensive resource qubit $\varphi$.

One natural follow-up is whether something resembling
recovery circuits exist for larger stabilizer circuits
that is distinct from stabilizer error correcting code
procedures.  This question has been answered to an extent
for the Clifford+$T$ gate set in \cite{brs.prob.cir,
brs.syn.rus, paetz.svore.rus}, where we can treat
$\ket{\psi} = HP^{\dagger}\ket{H_0}$ to perform a
non-Clifford $\pi/4$ phase rotation
\begin{align}
T = \begin{bmatrix} 1 &  0 \\ 0 & e^{i\frac{\pi}{4}} \end{bmatrix}.
\end{align}
The goal in \cite{brs.prob.cir, brs.syn.rus, paetz.svore.rus}
uses a multiqubit circuit of Clifford+$T$ gates to approximate
an arbitrary single qubit unitary $U$ up to some error $\epsilon$.
If the measurements are unfavorable, then there is a backup
operation that either returns the qubits to the initial
state, or directly tries to approximate $U$ using a secondary
circuit.  It is worth investigating whether there exists conditions
that enable larger stabilizer circuits to exhibit the recovery
feature we demonstrated here on general $\ket{\psi}$
resources.

Another direction that we may pursue is a more detailed
and thorough examination of the depth $k$ protocol.
In particular, there is an optimal number of circuits
to employ that uses the fewest number of resources in
expectation on each invocation. As we saw in Section
\ref{sec:expmts}, the behavior of our protocol is akin
to that of a (possibly non-uniform) random walk.  This
modeling of probabilistic circuits is nothing new
(see \cite{brs.prob.cir, dcs, jones.qchem}).
One matter we need to keep in mind is the costs
of attaining qubits $\varphi$ and $\ket{\psi}$.
The amount of work that went into preparing
$\varphi$ should exceed that of $\ket{\psi}$ in order
for the recovery to be cost effective, which stems
from the fact that we need a copy of $\ket{\psi}$ to
operate each circuit. The random walk techniques in
\cite{kemeny1960finite} should also prove useful for
gathering a more precise cost estimate. From there, we
could gain better insight into the overall capabilities
and limitations of stabilizer circuits acting on
non-stabilizer input.

\nocite{*}
\bibliographystyle{IEEEannot}
\bibliography{two-cir-rec}

\appendix
\section{Proof of Lemma~\ref{lemma:cir-outputs-all}}
\label{app:proof-eq-cir}

Similar to a single qubit, a two-qubit density matrix $\rho$ can
be expressed as a real combination of two-qubit Pauli operators
$\sigma_{jk} = \sigma_j \otimes \sigma_k$, where $\sigma_0 = I$,
$\sigma_1 = X$, $\sigma_2 = Y$, and $\sigma_3 = Z$
e.g.\ $\sigma_{13} = X \otimes Z$. We omit the tensor product and
use $\sigma_{jk}$ for notational reasons.  We define
$\mathcal{P}_{\pm} = \{ \pm \sigma_{jk} \mid j \ne 0
\text{ and } k \ne 0\}$ to be a set of nontrivial two-qubit Pauli
operators.

To prove Lemma~\ref{lemma:cir-outputs-all}, we start by
rewriting Equation~\ref{eq:cliff-eq-cir} in Definition
\ref{def:equiv-cir} as
\begin{align}
\label{eq:cliff-eq-proj}
C_1  \Pi_1 \rho \Pi_1 C^{\dagger}_1 & =
(G \otimes I)C_2 \Pi_2 \rho \Pi_2 C^{\dagger}_2 (G^{\dagger} \otimes I)
\end{align}
where $\Pi_1 = C^{\dagger}_1(I \otimes \ketbra{b_1}{b_1}) C_1$
and $\Pi_2 = C^{\dagger}_2(I \otimes \ketbra{b_2}{b_2}) C_2$
are projection operators.  Reichardt \cite{reichardt.msd.2006}
showed that Equation~\ref{eq:cliff-eq-proj} holds for some single
qubit Clifford $G$ on all states $\rho$ if $\Pi_1 = \Pi_2$.  In
our two-qubit case, there are only $30$ cases of $\Pi_1 = \Pi_2$.
We make some refinements here to make the ideas in
\cite{reichardt.msd.2006} a little more digestable in our
notation.

\begin{lemma}
\label{lemma:clifford-maps}
Let $(C_1, b_1)$ and $(C_2, b_2)$ be postselected two-to-one
stabilizer circuits.  If
\begin{align}
C^{\dagger}_1(I \otimes \ketbra{b_1}{b_1}) C_1 = \Pi =
C^{\dagger}_2(I \otimes \ketbra{b_2}{b_2}) C_2,
\end{align}
then $(C_1, b_1) \sim (C_2, b_2)$.
\end{lemma}
\begin{proof}
Note that $2(I \otimes \ket{b_j}\bra{b_j}) =
\sigma_{00} + (-1)^{b_j} \sigma_{03}$.
Let $2\Pi = \sigma_{00} + \lambda_{03}$, where
$\lambda_{03} \in \mathcal{P}_{\pm}$, and let
$\lambda_{10}$, $\lambda_{20}$, $\lambda_{30} \in \mathcal{P}_{\pm}$
be two-qubit Pauli operators such that
$[\lambda_{03}, \lambda_{10}] =
[\lambda_{03}, \lambda_{20}] =
[\lambda_{03}, \lambda_{30}] = 0$ and
$i\lambda_{30} = \lambda_{10}\lambda_{20}$.
Let $\rho$ be a two-qubit state.  Then
\begin{align}
\Pi \rho \Pi = \frac{1}{8} \left(
   w \sigma_{00} + w \lambda_{03} +
x \lambda_{10} + x \lambda_{13} +
y \lambda_{20} + y \lambda_{23} +
z \lambda_{30} + z \lambda_{33} \right)
\end{align}
where $\lambda_{k3} = \lambda_{03}\lambda_{k0}$ and
$x = \mathrm{Tr}((\lambda_{10} + \lambda_{13})\rho)$. The
coefficients $w$, $y$, $z$ are determined similarly with
$\sigma_{00} + \lambda_{03}$, $\lambda_{20}+\lambda_{23}$,
and $\lambda_{30} + \lambda_{33}$, respectively.
Our starting condition $C_j \lambda_{03} C^{\dagger}_j =
(-1)^{b_j}\sigma_{03}$ implies
\begin{align}
C_j\lambda_{10} C^{\dagger}_j,
C_j\lambda_{20} C^{\dagger}_j \in
  \{\,   \sigma_{10}&, (-1)^{b_j}\sigma_{13},
        -\sigma_{10}, (-1)^{b_j+1}\sigma_{13}, \nonumber \\
         \sigma_{20}&, (-1)^{b_j}\sigma_{23},
        -\sigma_{20}, (-1)^{b_j+1}\sigma_{23}, \nonumber \\
         \sigma_{30}&, (-1)^{b_j}\sigma_{33},
        -\sigma_{30}, (-1)^{b_j+1}\sigma_{33} \,\}.
\end{align}
This means there are single qubit Clifford gates $G_j$
to permute the operators in a way that
\begin{align}
(G_j\otimes I)C_j \lambda_{10} C^{\dagger}_j
             (G^{\dagger}_j\otimes I) & \in
      \{\,  \sigma_{10}, (-1)^{b_j}\sigma_{13} \,\} \\
(G_j\otimes I)C_j \lambda_{20} C^{\dagger}_j
             (G^{\dagger}_j\otimes I) & \in
      \{\,  \sigma_{20}, (-1)^{b_j}\sigma_{23} \,\}.
\end{align}
The value of $(G_j\otimes I)C_j \lambda_{30}
C^{\dagger}_j (G^{\dagger}_j\otimes I)$ is fixed given
the other two.  Our unnormalized post-measurement states
$\rho'_j = (G_j \otimes I)C_j\Pi \rho \Pi
C^{\dagger}_j(G^{\dagger}_j \otimes I)$ are now
\begin{align}
\rho'_j = \frac{1}{4}
            \left( w I + x X + y Y + z Z \right) \otimes
            \ket{b_j}\bra{b_j}.
\end{align}
The first qubit of $\rho'_1$ and $\rho'_2$ are
the same after $G_1$ and $G_2$.  Therefore
$(C_1, b_1) \sim (C_2, b_2)$.
\end{proof}

We now have the tools available to prove Lemma
\ref{lemma:cir-outputs-all}.  Note that a Clifford
equivalence $(C_1, b_1) \sim (C_2, b_2)$ is invariant
with respect to Clifford circuits that execute prior
to circuits $C_1$ and $C_2$ i.e.\ $(C_1, b_1) \sim
(C_2, b_2)$ if and only if $(C_1U, b_1) \sim (C_2U, b_2)$
for any Clifford unitary $U$.

\begin{proof}
We partition the $15$ Pauli operators $\sigma_{jk}$ into the following sets:
\begin{align}
\mathcal{P}_A &= \{\sigma_{jk}\mid j,k \in \{1,2,3\} \}, &
\mathcal{P}_B &= \{\sigma_{01}, \sigma_{02}, \sigma_{03}\}, &
\mathcal{P}_C &= \{\sigma_{10}, \sigma_{20}, \sigma_{30}\}.
\end{align}

We look at $\sigma_{33}$ first.  Suppose there is a bit $b'$
such that $C\sigma_{33}C^{\dagger} = (-1)^{b'}\sigma_{03}$.
For readability, set $C' = \cnot$.
Knowing $C'\sigma_{33}C'^{\dagger} = \sigma_{03}$, we obtain
$(C, b) \sim (\cnot, b + b'\ \mathrm{mod}\ 2)$ from
Lemma~\ref{lemma:clifford-maps}.  For the remaining
$\sigma_{jk} \in \mathcal{P}_A$, suppose
$C\sigma_{jk}C^{\dagger} = \pm \sigma_{03}$.
Choose single qubit Clifford
gates $G_1$ and $G_2$ such that $(G_1 \otimes G_2)\sigma_{jk}
(G^{\dagger}_1 \otimes G^{\dagger}_2) = \sigma_{33}$.
Define $C'' = C(G^{\dagger}_1 \otimes G^{\dagger}_2)$.
Then $C''\sigma_{33}C''^{\dagger} = (-1)^{b'} \sigma_{03}$
for some $b'$.  The rest follows from previous arguments
to conclude $(C''(G_1 \otimes G_2), b) =
(C, b) \sim (\cnot(G_1 \otimes G_2), b + b'\ \mathrm{mod}\ 2)$.

For the operator $\sigma_{03} \in \mathcal{P}_B$, assume
$C\sigma_{03}C^{\dagger} = (-1)^{b'}\sigma_{03}$. Then
$(C, b) \sim (\sigma_{00}, b + b'\ \mathrm{mod}\ 2)$.
Coverage of the other five from $\mathcal{P}_B$ and
$\mathcal{P}_C$ is similar to the above.

To finish, suppose $(C, b) \sim (I \otimes G, b+b'\ \mathrm{mod}\ 2)$,
where $G$ is a single qubit Clifford gate.
If $b + b'\ \mathrm{mod}\ 2 = 1$, then
$(C, b) \sim (I \otimes G, 1)\equiv (I \otimes XG, 0)$.
The same applies when $(C, b) \sim ((I\otimes G)\swap, 1)$.
If $(C, b) \sim (\cnot(G_1 \otimes G_2), 1)$, then we include
$(I \otimes X)\cnot(G_1 \otimes G_2) = \cnot(G_1 \otimes XG_2)$.
The other case $b+b'\ \mathrm{mod}\ 2 = 0$ follows
directly from Lemma~\ref{lemma:clifford-maps},
\end{proof}

\section{Additional Material on Recovery Circuits}
\label{app:extra}
We may use the following to help us determine
when two recovery circuits are Clifford equivalent.
In particular, it dispels concerns that there may
be two recovery circuits where one has a better
chance of succeeding than the other.  We use the
same notation for two-qubit Paulis $\sigma_{jk}$ and
$\mathcal{P}_{\pm}$ as in Appendix~\ref{app:proof-eq-cir}.

\begin{lemma}
\label{lemma:rec-cir-equiv}
Let $(C_1, b_1)$ be a recovery circuit of an interacting
postselected circuit $(C, b)$.  If $(C_2, b_2)$ is also a
recovery circuit of $(C, b)$, then
$C^{\dagger}_1 (I \otimes \ketbra{b_1}{b_1})C_1 =
C^{\dagger}_2 (I \otimes \ketbra{b_2}{b_2})C_2$.
\end{lemma}
\begin{proof}
It is easier to prove the contrapositive. Specifically, we show
the recovery from $(C_2, b_2)$ will fail on one particular pair
of qubits $\varphi_2$ and $\ket{\psi}$, although many exists
that are equally as good.  Suppose
$\Pi_2 = C^{\dagger}_2 (I \otimes \ketbra{b_2}{b_2})C_2$.
Let $2\Pi_2 = \sigma_{00} + \lambda_{03}$, where
$\lambda_{03} \in \mathcal{P}_{\pm}$, and let
$\lambda_{30}$ and $\lambda_{33}$ be two-qubit Pauli
operators from $\mathcal{P}_{\pm}$ such that
$[\lambda_{03}, \lambda_{30}] = 0$ and
$\lambda_{03} = \lambda_{30}\lambda_{33}$.
The qubits $\varphi_2$ and $\ket{\psi}$ we choose shall
have Bloch vectors
\begin{align}
  \varphi_2\colon
  \left( x_2, y_2, z_2 \right) &=
  \left( \sqrt{\frac{2}{17}},
         \sqrt{\frac{5}{17}},
         \sqrt{\frac{10}{17}} \right), &
  \ket{\psi}\colon
  \left( x, y, z \right) &=
  \left( \sqrt{\frac{1}{11}},
         \sqrt{\frac{3}{11}},
         \sqrt{\frac{7}{11}} \right).
\end{align}
Let $\varphi_1$ be a qubit so that
$\varphi_2 = \Phi_{1-b}(C, \varphi_1 \otimes\psi)$.
Let $\varphi'_1 = \Phi_{b_2}(C_2, \varphi_2 \otimes\psi)$.

To prove the recovery by $(C_2, b_2)$ will fail, we merely need to
verify that the Bloch vectors from all $18$ choices of $\lambda_{03}$
are different, which implies $\varphi'_1 \ne \varphi_1$ whenever
$C^{\dagger}_1 (I \otimes \ketbra{b_1}{b_1})C_1 \ne \Pi_2$. We
track the coefficients $a_{jk} = \mathrm{Tr}(\lambda_{jk}
(\varphi_2 \otimes \psi))$. Then
\begin{align}
\mathrm{Tr}\left( \Pi_2 \left( \varphi_2 \otimes
                              \psi \right)
                  \Pi_2 \right)
  &= \frac{1 + a_{03}}{2}, &
\mathrm{Tr}\left( \lambda_{30} \Pi_2
                  \left( \varphi_2 \otimes
                         \psi \right)
                  \Pi_2 \right)
  &= \frac{a_{30} + a_{33}}{2},
\end{align}
yielding $v = \frac{a_{30} + a_{33}}{1 + a_{03}}$ as a Bloch vector
component of $\rho'_1$.  The most convenient choices for $\lambda_{30}$
and $\lambda_{33}$ are tensor products with one identity e.g.\
$\lambda_{03} = -\sigma_{33}$, $\lambda_{30} = \sigma_{30}$,
$\lambda_{33} = -\sigma_{03}$, and $\lambda_{03} = \sigma_{11}$,
$\lambda_{30} = \sigma_{10}$, $\lambda_{33} = \sigma_{01}$, which
means that $a_{03} = a_{30}a_{33}$. If we look at the coefficients
from the first example with $\lambda_{03} = -\sigma_{33}$, then
$a_{30} = z_2$ and $a_{33} = -z$. We get the following components
for each possibility of $\lambda_{03}$:
\begin{align*}
\begin{array}{cc}
  \begin{array}{|c|c|c|c|c|c|c|} \hline
  \lambda_{03} & a_{03} & \lambda_{30} & a_{30} &
    \lambda_{33} & a_{33} & v \\ \hline
  \sigma_{11} & x_2x & \sigma_{10} & x_2 & \sigma_{01} & x & 0.5841 \\ \hline
  \sigma_{12} & x_2y & \sigma_{10} & x_2 & \sigma_{02} & y & 0.7338 \\ \hline
  \sigma_{13} & x_2z & \sigma_{10} & x_2 & \sigma_{03} & z & 0.8957 \\ \hline
  \sigma_{21} & y_2y & \sigma_{20} & y_2 & \sigma_{01} & x & 0.7252 \\ \hline
  \sigma_{22} & y_2y & \sigma_{20} & y_2 & \sigma_{02} & y & 0.8296 \\ \hline
  \sigma_{23} & y_2z & \sigma_{20} & y_2 & \sigma_{03} & z & 0.9354 \\ \hline
  \sigma_{31} & z_2x & \sigma_{30} & z_2 & \sigma_{01} & x & 0.8678 \\ \hline
  \sigma_{32} & z_2y & \sigma_{30} & z_2 & \sigma_{02} & y & 0.9205 \\ \hline
  \sigma_{33} & z_2z & \sigma_{30} & z_2 & \sigma_{03} & z & 0.9708 \\ \hline
  \end{array} &
  \begin{array}{|c|c|c|c|c|c|c|} \hline
  \lambda_{03} & a_{03} & \lambda_{30} & a_{30} &
    \lambda_{33} & a_{33} & v \\ \hline
  -\sigma_{11} & -x_2x & \sigma_{10} & x_2 & -\sigma_{01} & -x &  0.0463 \\ \hline
  -\sigma_{12} & -x_2y & \sigma_{10} & x_2 & -\sigma_{02} & -y & -0.2183 \\ \hline
  -\sigma_{13} & -x_2z & \sigma_{10} & x_2 & -\sigma_{03} & -z & -0.6260 \\ \hline
  -\sigma_{21} & -y_2y & \sigma_{20} & y_2 & -\sigma_{01} & -x &  0.2879 \\ \hline
  -\sigma_{22} & -y_2y & \sigma_{20} & y_2 & -\sigma_{02} & -y &  0.0280 \\ \hline
  -\sigma_{23} & -y_2z & \sigma_{20} & y_2 & -\sigma_{03} & -z & -0.4501 \\ \hline
  -\sigma_{31} & -z_2x & \sigma_{30} & z_2 & -\sigma_{01} & -x &  0.6055 \\ \hline
  -\sigma_{32} & -z_2y & \sigma_{30} & z_2 & -\sigma_{02} & -y &  0.4083 \\ \hline
  -\sigma_{33} & -z_2z & \sigma_{30} & z_2 & -\sigma_{03} & -z & -0.0792 \\ \hline
  \end{array}
\end{array}
\end{align*}
Neither are any of the values $v$ the same if we multiple each
one by $-1$, which may come about from an application of a
single qubit Pauli on the output.  Thus our statement holds.
\end{proof}

\end{document}